\documentclass[runningheads]{llncs}
\usepackage[utf8]{inputenc}
\usepackage{graphicx}
\usepackage{amssymb}
\usepackage{amsmath}
\usepackage{latexsym}
\usepackage{todonotes}
\usepackage{color}
\usepackage{subfigure}
\usepackage{wrapfig}

\usepackage{multirow}
\usepackage{arydshln}
\usepackage{thm-restate}
\usepackage{lineno}

\spnewtheorem*{sketch}{Sketch of proof}{\itshape}{\rmfamily}
\renewcommand{\deg}{\mathsf{deg}}
\renewcommand{\qed}{\hfill $\square$}

\newcommand{\remove}[1]{}

\begin{document}

\newtheorem{observation}{Observation}
\newtheorem{myclaim}{Claim}

\title{ $(k,p)$-Planarity: A Relaxation of \\Hybrid Planarity}

\author{
Emilio Di Giacomo\inst{1}
\and
William J. Lenhart\inst{2}
\and
Giuseppe Liotta\inst{1}
\and \\
Timothy W. Randolph\inst{3}
\and
Alessandra Tappini\inst{1}
}

\date{}

\institute{
Universit\`a degli Studi di Perugia, Italy\\
\email{\{emilio.digiacomo,giuseppe.liotta\}@unipg.it}
\email{alessandra.tappini@studenti.unipg.it}
\and
Williams College, USA\\
\email{wlenhart@williams.edu}
\and
Columbia University, USA\\
\email{t.randolph@columbia.edu}
}

\maketitle

% \linenumbers

\begin{abstract}
We present a new  model for hybrid planarity that relaxes existing hybrid representations. A graph $G = (V,E)$ is $(k,p)$-planar if $V$ can be partitioned into clusters of size at most $k$ such that $G$ admits a drawing where: (i) each cluster is associated with a closed, bounded planar region, called a \emph{cluster region}; (ii) cluster regions are pairwise disjoint, (iii) each vertex $v \in V$ is identified with at most $p$ distinct points, called \emph{ports}, on the boundary of its cluster region; (iv) each inter-cluster edge $(u,v) \in E$ is identified with a Jordan arc connecting a port of $u$ to a port of $v$; (v) inter-cluster edges do not cross or intersect cluster regions except at their endpoints. We first tightly bound the number of edges in a $(k,p)$-planar graph with $p<k$. We then prove that $(4,1)$-planarity testing and $(2,2)$-planarity testing are NP-complete problems. Finally, we prove that neither the class of $(2,2)$-planar graphs nor the class of $1$-planar graphs contains the other, indicating that the $(k,p)$-planar graphs are a large and novel class.

\keywords{$(k,p)$-planarity \and hybrid representations \and cluster graphs}

\end{abstract}

%%%%%%%

\section{Introduction}
Visualization of non-planar graphs is one of the most studied graph-drawing problems in recent years. In this context, an emerging topic is hybrid representations (see, e.g.,~\cite{DBLP:conf/gd/AngeliniLBFPR15,Batagelj_Visual_2011,DBLP:conf/gd/LozzoBFP16,dlpt-ntptsc-17,hfm-dhvsn-07}).
A hybrid representation simplifies the visual analysis of a non-planar graph by adopting different visualization paradigms for different portions of the graph. The graph is divided into (typically dense) subgraphs called \emph{clusters} which are restricted to limited regions of the plane. Edges between vertices in the same cluster are called \emph{intra-cluster} edges, and edges between vertices in different clusters are called \emph{inter-cluster} edges. Inter-cluster edges are represented according to the classical node-link graph drawing paradigm, while the clusters and their intra-cluster edges are represented by adopting alternative paradigms. A hybrid representation thus reduces the number of inter-cluster edges and the visual complexity of much of the drawing at the cost of creating cluster regions of high visual complexity. As a result, a hybrid representation provides an easy to read overview of the graph structure and it admits a ``drill-down'' approach when a more detailed analysis of some of its clusters is needed.

%\setlength\intextsep{0pt}
%\setlength{\columnsep}{10pt}%
%\begin{wrapfigure}{r}{0.32\textwidth}
%	\centering
%	\subfigure[] {\includegraphics[width=0.29\textwidth]{NodeTrix} \label{fi:NodeTrix}}
%%	\hfil
%	\subfigure[] {\includegraphics[width=0.29\textwidth]{intersection-link} \label{fi:intersection-link}}
%	\caption{(a) A NodeTrix representation of a $3$-clique and a corresponding $(3,4)$ representation. (b) An intersection-link representation of a $3$-clique and a corresponding $(3,2)$ representation.} \label{fi:NodeTrix-intersectionlink}
%    \label{NodeTrix_intersection-link}
%\end{wrapfigure}

Different representation paradigms for clusters give rise to different types of hybrid representations. For example, Angelini et al.~\cite{DBLP:conf/gd/AngeliniLBFPR15} introduce \emph{intersection-link} representations, where clusters are represented as intersection graphs of sets of rectangles, while Henry et al.~\cite{hfm-dhvsn-07} introduce \emph{NodeTrix} representations, where dense subgraphs are represented as adjacency matrices (see Fig.~\ref{fi:NodeTrix-intersectionlink}).
Batagelj et al. employ hybrid representations in the \emph{$(X,Y)$-clustering} model~\cite{Batagelj_Visual_2011}, where $Y$ and $X$ define the desired topological properties of the clusters and of the graph connecting the clusters, respectively. For instance, in a $(planar, k$-$clique)$-clustering of a graph each cluster is a $k$-clique and the graph obtained by contracting each cluster into a single node (called the \emph{graph of clusters}) is planar.
%Batagelj et al.~\cite{Batagelj_Visual_2011} combine a node-link representation of the graph of clusters (the ``$X$ part'') with alternative representations (adjacency matrices, circular layout, etc.) for the clusters (the  ``$Y$ part'').
Given a graph $G$ and a hybrid representation paradigm $\mathcal{P}$, the \emph{hybrid planarity} problem asks whether $G$ can be represented according to $\mathcal{P}$ with no inter-cluster edge crossings. Variants of the problem may or may not assume that the clustering is given as part of the input.
%Da Lozzo et al.~\cite{DBLP:conf/gd/LozzoBFP16} prove that the problem of testing whether a clustered graph admits a planar NodeTrix representation is NP-complete, while Di Giacomo et al.~\cite{dlpt-ntptsc-17} study NodeTrix planarity when the maximum size of the clusters is bounded by a constant $k$.
%They prove that if the graph of clusters is a partial $2$-tree, the problem can be solved in $O(k^{3k+\frac{3}{2}}n^3)$ time, while for general graphs of clusters the problem is NP-complete for $k \geq 3$ and can be solved in $O(n^3)$ time for $k=2$. These results apply in the \emph{fixed sides scenario} (i.e. for each matrix it is specified which edges must be incident on each of the four sides), while in the \emph{free sides scenario} the problem remains NP-complete for any $k>4$.

\begin{figure}[tb]
	\centering
	\subfigure[] {\includegraphics[width=0.35\textwidth]{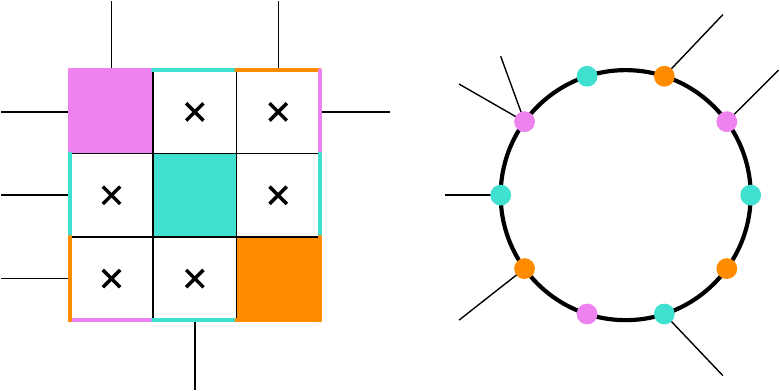} \label{fi:NodeTrix}}
	\hfil
	\subfigure[] {\includegraphics[width=0.35\textwidth]{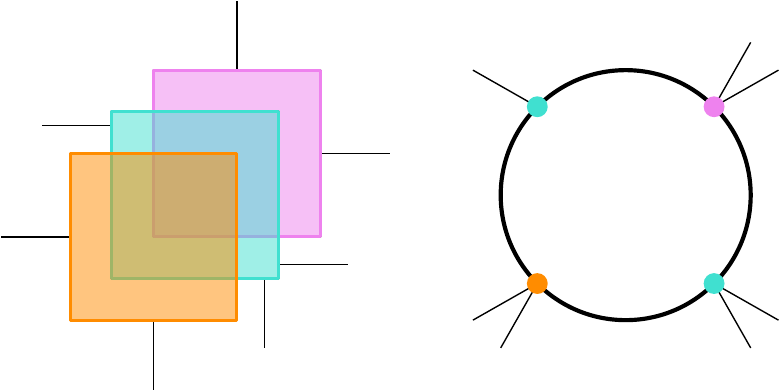} \label{fi:intersection-link}}
	\caption{(a) A NodeTrix representation of a $3$-clique and a corresponding $(3,4)$ representation. (b) An intersection-link representation of a $3$-clique and a corresponding $(3,2)$ representation.} \label{fi:NodeTrix-intersectionlink}
    \label{NodeTrix_intersection-link}
\end{figure}

In this paper, we present a general hybrid representation paradigm that relaxes the described hybrid paradigms. Given a graph $G = (V,E)$, a \emph{(k,p) representation} $\Gamma$ of $G$ is a hybrid representation in which: (i) each cluster of $G$ contains at most $k$ vertices and is identified with a closed, bounded planar region; (ii) cluster regions are pairwise disjoint, (iii) each vertex $v \in V$ is represented by at most $p$ distinct points, called \emph{ports}, on the boundary of its cluster region; (iv) each inter-cluster edge $(u,v) \in E$ is represented by a Jordan arc connecting a port of $u$ to a port of $v$. A $(k,p)$ representation is \emph{$(k,p)$-planar} if edge curves do not cross and do not intersect cluster regions except at their endpoints. We say that a graph $G$ is \emph{$(k,p)$-planar} if it can be clustered so that it admits a $(k,p)$-planar representation.

The definition of a $(k,p)$ representation leaves the representation of clusters and intra-cluster edges intentionally unspecified. It is thus a relaxation of hybrid representation paradigms where the number of ports used by the inter-cluster edges depends on the geometry of the cluster regions. For example, in a NodeTrix representation, the squared boundary of each matrix allows four ports for every vertex except for the vertex in the first row/column of the matrix and the vertex in the last row/column of the matrix, which both have only three ports. Hence, a NodeTrix representation can be regarded as a constrained $(k,4)$ representation (four ports for every vertex except for two, the vertices appear in the order imposed by the matrix); see Fig.~\ref{fi:NodeTrix}. Similarly, a $(k,2)$ representation relaxes an intersection-link representation with clusters represented as isothetic unit squares with their upper-left corners along a common line with slope $1$; see Fig.~\ref{fi:intersection-link}. We also remark that the use of different ports to represent a vertex can be regarded as an example of \emph{vertex splitting}~\cite{DBLP:conf/gd/EadesM95,Eppstein2018}; however, while in the papers that use vertex splitting to remove crossings the multiple copies of each vertex can be placed anywhere in the drawing, in our model they are forced to lay within the boundary of the same cluster region.

The results of this paper are the following:

\begin{itemize}
	\item
    %We study the maximum edge density of $(k,p)$-planar graphs .
In Section~\ref{se:density}, we give an upper bound on the edge density of a $(k,p)$-planar graph and  prove that this bound is tight for $p < k$.
	
%	\item We study the complexity of recognizing $(k,p)$-planar graphs (Section~\ref{se:testing}). We prove that  $(4,1)$-planarity testing and $(2,2)$-planarity testing are both NP-complete problems.
\item In Section~\ref{se:testing}, we observe that the class of $(4,1)$-planar graphs coincides with the class of IC-planar graphs, from which the NP-completeness of testing $(4,1)$-planarity follows.
We then prove that testing $(2,2)$-planarity is NP-complete. These results imply that computing the minimum $k$ such that a graph is $(k,p)$-planar is NP-hard for both $p=1$ and $p=2$.
Recall that a graph is \emph{1-planar} if it admits a drawing where every edge is crossed at most once, and that an \emph{IC-planar} graph is a $1$-planar graph that admits a drawing where no two pairs of crossing edges share a vertex.

\item The NP-completeness of the $(2,2)$-planarity testing problem naturally suggests to further investigate the combinatorial properties of $(2,2)$-planar graphs. In Section~\ref{se:1-planar}, we ask whether every $1$-planar graph admits a $(2,2)$-planar representation (see, e.g. Fig.~\ref{fi:1-planar-replacement}). We prove the existence of $1$-planar graphs that are not $(2,2)$-planar and of $(2,2)$-planar graphs that are not $1$-planar. We also give a sufficient condition for $1$-planar graphs to be $(2,2)$-planar.

%\item In Section~\ref{se:1-planar}, we further explore the combinatorial properties of $(2,2)$-planar graphs, demonstrating the existence of 1-planar graphs that are not $(2,2)$-planar and $(2,2)$-planar graphs that are not 1-planar. We also give a sufficient condition for $1$-planar graphs to be $(2,2)$-planar.

\end{itemize}

% Preliminary definitions are in Section~\ref{se:preliminaries}, while some open problems are discussed in Section~\ref{se:open-problems}.
Sections of certain proofs are removed to the appendix. These statements are marked with [*].

\section{Edge Density of $(k,p)$-Planar Graphs} \label{se:density}

In this section we give a tight bound on the number of edges of a $(k,p)$-planar graph when $p<k$. First, given a $(k,p)$-planar representation $\Gamma$, we define a \emph{skeleton} of $\Gamma$ to be a planar drawing $\Gamma_S$ obtained by the following transformation. We first replace each port in $\Gamma$ with a vertex. Each cluster region $R_i$ of $\Gamma$ is now an empty convex space surrounded by up to $kp$ vertices. We connect these vertices in a cycle and triangulate the interior. For our purposes any triangulation is equivalent. The resulting representation is $\Gamma_S$. Figure~\ref{fi:cluster-skeletonized-b} illustrates a skeleton of the $(2,2)$-planar representation of Fig.~\ref{fi:cluster-skeletonized-a}.
%We say that a graph is \emph{maximal $(k,p)$-planar} if the addition of any edge results in a graph that is not $(k,p)$-planar.

%The following property immediately derives from the fact that every cluster can be a clique.
%
%\begin{property} \label{pr:intra-cluster}
%	For any positive integers $k$ and $p$, a $(k,p)$-planar graph with $N$ clusters has at most $\frac{k(k-1)N}{2}$ intra-cluster edges.
%\end{property}

\begin{figure}[tb]
	\centering
	\subfigure[] {\includegraphics[width=0.28\textwidth,page=4]{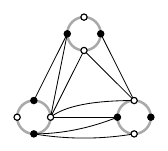} \label{fi:cluster-skeletonized-a}}
	\hfil
	\subfigure[] {\includegraphics[width=0.28\textwidth,page=2]{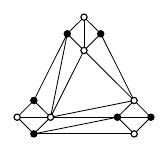} \label{fi:cluster-skeletonized-b}}
	\caption{(a) A $(2,2)$-planar representation $\Gamma$ of a graph $G$; (b) A skeleton $\Gamma_S$ of $\Gamma$.} \label{fi:cluster-skeletonized}
\end{figure}

\begin{theorem}\label{thm:edgebound2}\emph{\textbf{[*]}}
	Let $G$ be a $(k,p)$-planar graph with $n$ vertices. $G$ has $m \leq n(p+\frac{3}{k}+\frac{k}{2}-\frac{1}{2})-6$ edges. This bound is tight for any positive integers $k$, $p$ and $n$ such that $p<k$ and $n=N \cdot k$, where $N>2$.
\end{theorem}
%\begin{theorem} \label{thm_edgebound2}
%	Let $G$ be a $(k,p)$-planar graph with $n$ vertices. $G$ has $m \leq n(p+\frac{3}{k}+\frac{k}{2}-\frac{1}{2})-6$ edges.
%\end{theorem}
\begin{proof}
	Let $\Gamma$ be a $(k,p)$-planar representation of $G$ and let $N$ be the number of clusters of $G$. As each cluster contains at most $k$ vertices, $G$ has at most $N \cdot \frac{k(k-1)}{2}$ intra-cluster edges.

    Let $R_i$ be a cluster region in $\Gamma$ with $p_i$ ports in total. Let $\Gamma_S$ be a skeleton of $\Gamma$, and let $n_S$ and $m_S$ denote the number of vertices and the number of edges of $\Gamma_S$, respectively. When $\Gamma_S$ is created, $R_i$ is replaced with $p_i$ vertices and $2p_i-3$ edges if $p_i > 1$, or 0 edges if $p_i=1$. Letting $m_{inter}$ be the number of inter-cluster edges in $G$ and $s$ be the number of clusters in $G$ containing a single vertex,  we have,
    \begin{equation}
    m_S = m_{inter} + \sum_{i=1}^N (2p_i-3) + s.
	\end{equation}
In other words, the total number of edges in $\Gamma_S$ is equal to the number of inter-cluster edges in $G$ plus the number of edges added for each cluster.
	Note that $m_S \leq 3n_S - 6$, as $\Gamma_S$ is a planar drawing. As $\sum_{i=1}^N p_i = n_S$, rearranging generates $m_{inter} + 2n_S - 3N + s \leq 3n_S - 6$ and thus,
    \begin{equation} \label{eq:9}
	m_{inter} \leq n_S + 3N - 6 - s \leq N(kp+3) - 6.
	\end{equation}
As $m$ is equal to the sum of the number of inter-cluster and intra-cluster edges in $G$, we have
	\begin{equation}
    m \leq Nk(p + \frac{3}{k} + \frac{k}{2} - \frac{1}{2}) - 6.
    \end{equation}
%ince each cluster can have at most $\frac{k(k-1)}{2}$ intra-cluster edges, $G$ has at most $\frac{k(k-1)N}{2}$ intra-cluster edges. We now show that $\Gamma$ has at most $(kp+3)N-6$ inter-cluster edges.  Let $\Gamma_S$ be a skeleton of $\Gamma$. Denote by $n_S$ and $m_S$ the number of vertices and the number of edges of $\Gamma_S$. The value $m_S$ is equal to the number of inter-cluster edges $m_{inter}$ plus the number of edges added in place of each cluster region of $\Gamma$ to create $\Gamma_S$. If a cluster region $R_i$ has $p_i$ ports, connecting the ports with a cycle and triangulating the interior cycle creates $2p_i-3$ additional edges if $p_i > 1$, or $2p_i-2$ additional edges if $p_i=1$. Letting $s$ be the number of clusters of $\Gamma$ consisting of a single vertex, so that $2 p_1 - 2 = 2 p_1 - 3 + s$, we have that $m_S = m_{inter} + \sum_{i=1}^{N} (2p_i - 3) + s$ or, regrouping, $m_S = m_{inter} + 2 \sum_{i=1}^{N} p_i - 3N + s$. Note that $\sum_{i=1}^{N} p_i =n_S$, hence $m_S = m_{inter} + 2n_S - 3N + s$. Because $\Gamma_S$ is planar, we have $m_S \leq 3n_S - 6$ and therefore
%	\begin{equation} \label{eq:9}
%	m_{inter} + 2n_S - 3N + s \leq 3n_S - 6.
%	\end{equation}
%	Since $n_S \leq kpN$, we obtain $m_{inter} \leq (kp + 3)N - 6 - s \leq (kp + 3)N - 6$.
%	
%	To obtain the desired bound, we express the value of $N$ in terms of $n$ and $k$.

%%

If all clusters contain $k$ vertices, then $N = \frac{n}{k}$ and Theorem 1 holds. Appendix \ref{apx:edgebound2} completes the proof that $m \leq n(p + \frac{3}{k} + \frac{k}{2} - \frac{1}{2}) - 6$ in the case where some clusters contain fewer than $k$ vertices.

%Otherwise, we construct a sequence $\Gamma_0,\Gamma_1,\dots,\Gamma_N$ of $(k,p)$-planar representations such that $\Gamma_0 = \Gamma$ and $\Gamma_i$ is obtained from $\Gamma_{i-1}$ by removing $R_i$ if $R_i$ contains a single vertex or adding vertices to $R_i$ until $R_i$ contains $k$ vertices otherwise.

%If all clusters have size $k$, then $N=\frac{n}{k}$ and therefore $m \leq (kp + 3)N - 6 + \frac{k(k-1)}{2}N = n(p+\frac{3}{k}+\frac{k}{2}-\frac{1}{2})-6$. If the clusters have size at most $k$, we construct a sequence $\Gamma=\Gamma_0,\Gamma_1,\dots,\Gamma_N$ of $(k,p)$-planar representations so that $\Gamma_N$ has all clusters of size $k$ and each $\Gamma_i$ is obtained from $\Gamma_{i-1}$ by adding (or removing, if $|V_i| = 1$) vertices and edges to the cluster $V_i$ ($i=1,2,\dots,N$). We then prove the following claim, where $n_i$ and $m_i$ are the number of vertices and edges of $\Gamma_i$, respectively.
	
%	\begin{restatable}{myclaim}{claimdensity}\label{cl:one}
%		If $m_i \leq  n_i(p+\frac{3}{k}+\frac{k}{2}-\frac{1}{2})-6$ then $m_{i-1} \leq  n_{i-1}(p+\frac{3}{k}+\frac{k}{2}-\frac{1}{2})-6$.
%	\end{restatable}
	
%	Claim~\ref{cl:one} together with the fact that $m_N \leq  n_N(p+\frac{3}{k}+\frac{k}{2}-\frac{1}{2})-6$ (because $\Gamma_N$ has all clusters of size $k$) implies that $m_0 \leq n_0(p+\frac{3}{k}+\frac{k}{2}-\frac{1}{2})-6$. Since $n=n_0$ and $m \leq m_0$, the upper bound on the number of edges follows.
	
	In order to show that the bound is tight for $p<k$, we describe a $(k,p)$-planar representation $\Gamma_{k,p}$ with $N=\frac{n}{k}$ clusters and $(kp + 3)N - 6$ inter-cluster edges. $\Gamma_{k,p}$ is possible for any pair of positive integers $p$ and $k$ such that $p<k$ and for any $N > 2$. $\Gamma_{k,p}$ has $N$ clusters each with $k$ vertices and thus $kp$ ports. Let $R_1$ and $R_2$ be two cluster regions. We say that $R_1$ and $R_2$ are \emph{kp-connected} if they are connected by $kp+1$ edges as shown in Fig.~\ref{fi:kp-connected}(a). (Note that, since the number of inter-cluster edges between two $k$-clusters is at most $k^2$, we can create $kp+1$ edges between $R_1$ and $R_2$ only if $p<k$). More precisely, $R_1$, which we refer to as the \emph{small end} of the $kp$-connection, is connected by means of $p+1$ consecutive ports; the first $p$ ports have $k$ incident edges each, and the last port has an additional edge. $R_2$, which we refer to as the \emph{large end} of the $kp$-connection, is connected by means of $p(k-1)+1$ consecutive ports, each connected to one or two edges. Notice that, since we use $p(k-1)+1$ ports for the large end, $p+1$ for the small end and two ports can be shared by the two ends, each cluster region can be the small end of one $kp$-connection and the large end of another $kp$-connection. Thus, we can create a cycle with $N$ clusters as shown in Fig.~\ref{fi:kp-connected}(b). In the resulting representation there are two faces of degree $N$: One is the outer face and the other one is inside the cycle. By triangulating these two faces with $N-3$ edges for each face, we obtain the $(k,p)$-representation $\Gamma_{k,p}$. The number of inter-cluster edges of $\Gamma_{k,p}$ is thus $(kp+1)N+2N-6=(kp+3)N-6$. \qed
   %while the number of intra-cluster edges is $\frac{k(k-1)}{2}N$ (each cluster has size $k$). \qed
% Thus, $m=(kp + 3)\frac{n}{k} - 6+\frac{k(k-1)}{2}\frac{n}{k}=n(p+\frac{3}{k}+\frac{k}{2}-\frac{1}{2})-6$.\qed
\end{proof}

	\begin{figure}[tb]
		\centering
		\subfigure[]{\includegraphics[width=0.38\textwidth, page=5]{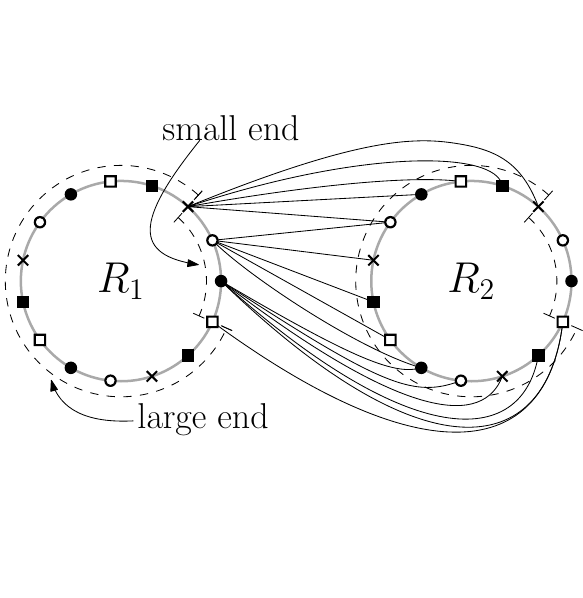}}
		\hfil
		\subfigure[]{\includegraphics[width=0.38\textwidth, page=6]{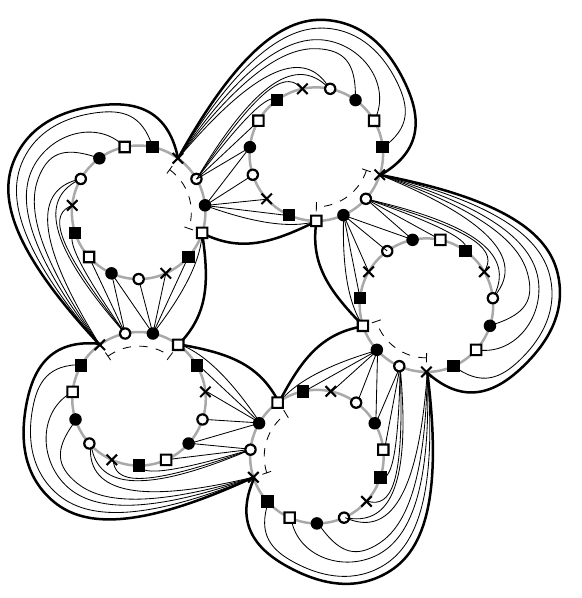}}	
		\caption{(a) A $kp$-connection of two cluster regions $R_1$ and $R_2$ ($k=5$, $p=3$). (b) A cycle of $N=5$ clusters; the bold edges highlight the two faces of degree $N$.} \label{fi:kp-connected}
	\end{figure}

\section{Recognition of $(k,p)$-Planar Graphs} \label{se:testing}

This section considers the problem of testing $(k,p)$-planarity for the cases in which $p=1$ and $p=2$.

%In this section, we study the $(k,p)$-planarity testing problem. It is not difficult to see that $(k,1)$-planarity for $k\leq 3$ is equivalent to standard planarity, while for $k=4$ it is equivalent to IC-planarity, whose test is known to be NP-complete~\cite{DBLP:journals/tcs/BrandenburgDEKL16}. We recall that \emph{IC-planar graphs} are $1$-planar graphs with a $1$-planar embedding such that no two pairs of crossing edges share a vertex~\cite{Zhang2013}. As a consequence of the NP-completeness of $(4,1)$-planarity, the problem of computing the minimum $k$ such that a graph is $(k,p)$-planar for $p=1$ is NP-hard.

%In this section we prove that the $(k,1)$-planarity testing problem can be solved in linear time for $k \leq 3$ and it is $NP$-complete for $k=4$, and that the $(2,2)$-planarity testing problem is $NP$-complete. As a consequence, the problem of computing the minimum $k$ such that a graph is $(k,p)$-planar is NP-hard both for $p=1$ and $p=2$. An \emph{IC-planar} graph is a $1$-planar graph that has a $1$-planar embedding such that no two pairs of crossing edges share a vertex~\cite{Zhang2013}.

\begin{restatable}{theorem}{theoremlinear}\label{th:$(3,1)$-linear}\emph{\textbf{[*]}}
	$(k,1)$-planarity testing can be performed in linear time for $k \leq 3$, and it is NP-complete for $k=4$.
\end{restatable}

\begin{proof}
The first part of Theorem \ref{th:$(3,1)$-linear} follows from the fact that the class of $(k,1)$-planar graphs coincides with the class of planar graphs for $k=1,2,3$. The second part follows from the fact that the $(4,1)$-planar graphs coincide with the IC-planar graphs~\cite{Zhang2013}. Testing IC-planarity is known to be NP-complete ~\cite{DBLP:journals/tcs/BrandenburgDEKL16}. Appendix B proves both equivalencies. \qed
\end{proof}

\begin{corollary}\label{co:hardness-p-1}
The problem of computing the minimum value of $k$ such that a graph is $(k,1)$-planar is NP-hard.
\end{corollary}

We now focus on the $(2,2)$-planarity testing problem, hereafter referred to as {\sc $(2,2)$-Planarity}. We show that {\sc $(2,2)$-Planarity} is NP-complete by a reduction from the NP-complete problem {\sc Planar Monotone 3-SAT}~\cite{DBLP:journals/ijcga/BergK12}.
%Also in this case, a consequence of this result is that minimizing $k$ such that a graph is $(k,2)$-planar is NP-hard.
We say that an instance of 3-SAT is \emph{monotone} if every clause consists solely of positive literals (a \emph{positive clause}) or solely of negative literals (a \emph{negative clause}). A \emph{rectilinear representation} of a 3-SAT instance is a planar drawing where each variable and clause is represented by a rectangle, all the variable rectangles are drawn along a horizontal line, and vertical segments connect clauses with their constituent variables. A rectilinear representation is \emph{monotone} if it corresponds to a monotone instance of planar 3-SAT where positive clauses are drawn above the variables and negative clauses are drawn below the variables, as shown in Fig.~\ref{fi:x_0}. Given a monotone rectilinear representation $\Phi$ corresponding to a boolean formula $F$, the problem {\sc Planar Monotone 3-SAT} asks if $F$ has a satisfying assignment.

\begin{figure}[tb]
	\centering
	\subfigure[]{\includegraphics[width=0.32\textwidth]{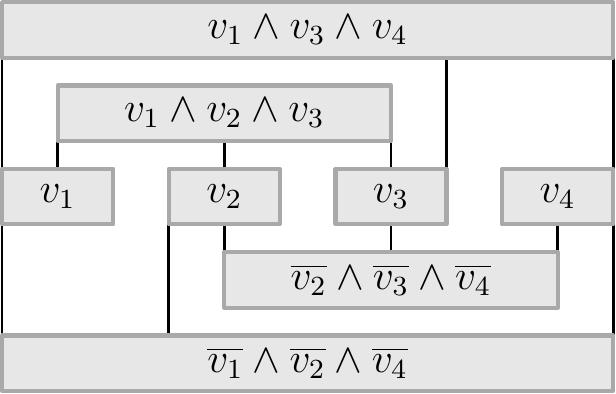}\label{fi:x_0}}
	\hfil
	\subfigure[]{\includegraphics[width=0.53\textwidth]{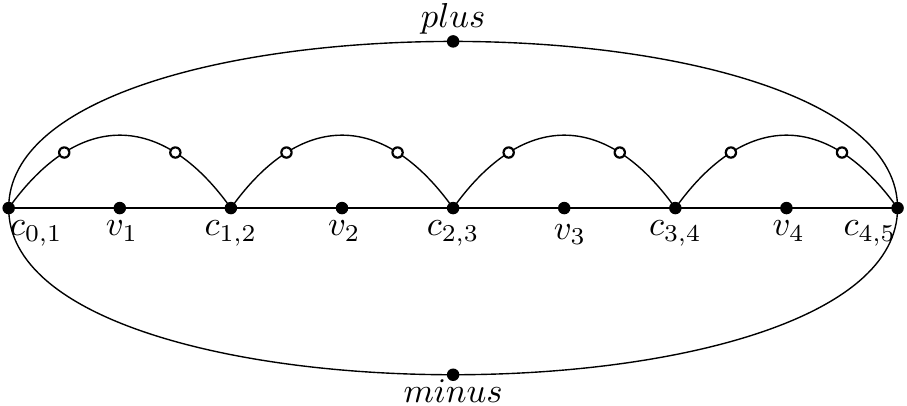}\label{fi:variable_cycle}}
	\caption{(a) A planar monotone representation of $\Phi_0$. (b) The variable cycle of $G_0$ and false literal boundaries.}
\end{figure}
	
We denote by $K^-_{8}$ the graph created by removing two adjacent edges from the complete graph $K_8$. In our reduction we
make use of the following transformation.  Let $v$ be a vertex of $G$. we replace $v$ with a copy of $K^-_{8}$ by identifying $v$ with the vertex of $K^-_{8}$ with degree $5$. After performing this operation we say that $v$ is a $K$-vertex. The following lemma, whose proof is in Appendix~\ref{apx:Kvertex}, states a useful property of the $K$-vertices.
	
	\begin{restatable}{lemma}{lemmaKvertex}\label{le:Kvertex}\emph{\textbf{[*]}}
%		Any $(2,2)$-planar representation of a $K$-vertex $v$ clusters its associated $K^-_{8}$ subgraph into four $2$-clusters.
		Let $v$ be a $K$-vertex of a graph $G$ and let $G'$ be the $K^-_{8}$ subgraph associated with $v$. In any $(2,2)$-planar representation of $G$, each vertex of $G'$ is clustered with another vertex in $G'$.
	\end{restatable}

\begin{restatable}{theorem}{theoremNPhardness} \label{thm:22_NP-hard}\emph{\textbf{[*]}}
	{\sc $(2,2)$-Planarity} is NP-complete.
\end{restatable}

\begin{proof}
{\sc $(2,2)$-Planarity} is trivially in NP. We prove the NP-hardness of {\sc $(2,2)$-Planarity} by reduction from {\sc Planar Monotone 3-SAT}. Given an instance $\Phi$ of {\sc Planar Monotone 3-SAT}, we construct a graph $G$ that is a {\sc Yes} instance of {\sc $(2,2)$-Planarity} if and only if $\Phi$ is a {\sc Yes} instance of {\sc Planar Monotone 3-SAT}. For convenience, figures show the construction of the graph $G_0$ corresponding to the {\sc Planar Monotone 3-SAT} instance $\Phi_0$ in Fig.~\ref{fi:x_0}. In figures, we represent $K$-vertices and their associated $K^-_{8}$ subgraphs with solid dots, while ordinary vertices are represented with hollow dots.

For each variable $v_i$ of $F$ (with $i=1,\dots,n$) create in $G$ a $K$-vertex $v_i$ and connect such $K$-vertices in a cycle, in the order implied by $\Phi$ (refer to Fig.~\ref{fi:variable_cycle}). Split each edge ($v_i$, $v_{i+1}$) of the cycle with a $K$-vertex $c_{i, i+1}$. Split the edge ($v_1$, $v_n$) with the vertices $c_{0, 1}$ and $c_{n, n+1}$. Finally, duplicate the edge $(c_{0, 1}, c_{n, n+1})$ and split the duplicated edges with the $K$-vertices $plus$ and $minus$. We refer to this subgraph as the \emph{variable cycle}. Given a variable $v_i$, let $p_i$ be the number of positive clauses and $q_i$ be the number of negative clauses of $F$ in which $v_i$ appears. For $1\leq i\leq n$, connect $c_{i-1, i}$ to $c_{i, i+1}$ with a path of ordinary vertices of length equal to $max(p_i, q_i)$. We refer to these paths as \emph{false literal boundaries.}

For each clause $C_j = (l_{j,1} \vee l_{j,2} \vee l_{j,3})$ in $F$, create a corresponding clause gadget in $G$. Create ordinary vertices $l_{j,1}$, $l_{j,2}$, $l_{j,3}$ and $open_j$, create a $K$-vertex $closed_j$, and add an edge between any pair of vertices, as in Fig.~\ref{fi:clause_gadget}. Observe that in any $(2,2)$-planar representation of a clause gadget, two of the four vertices $l_{j,1}, l_{j,2}, l_{j,3}$ and $open_j$ must be arranged in one cluster of size $2$. This is due to the fact that by Lemma~\ref{le:Kvertex}, $closed_j$ must be clustered within its $K^-_{8}$ subgraph. If $l_{j,1}, l_{j,2}, l_{j,3}$ and $open_j$ were all clustered separately, the graph of clusters of $G$ would contain a $K_5$ minor. Also, any $2$-clustering of a clause gadget in which a literal vertex is clustered with $open_j$ is $(2,2)$-planar, as shown in Fig.~\ref{fi:22_clause_gadget}.

\begin{figure}[tb]
	\centering
	\subfigure[] {\includegraphics[width=0.2\textwidth]{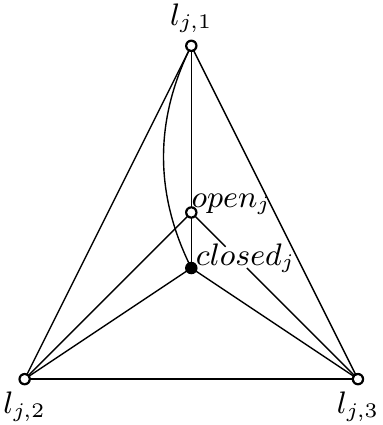} \label{fi:clause_gadget}}
	\hfil
	\subfigure[] {\includegraphics[width=0.2\textwidth]{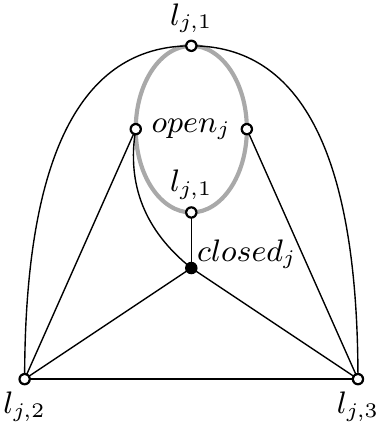} \label{fi:22_clause_gadget}}
	\hfil
	\subfigure[]{\includegraphics[width=0.45\textwidth]{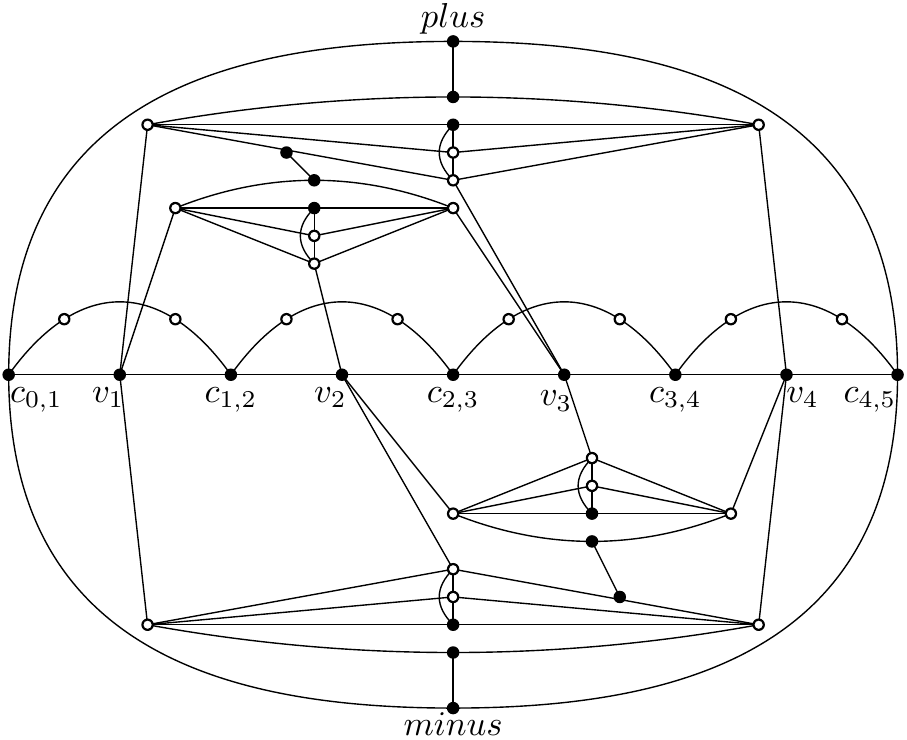} \label{fi:clause_forest}}
	\caption{ (a) A clause gadget $C_j$. (b) A $(2,2)$-planar representation of the clause gadget $C_j$. (c) The graph $G_0$.} \label{fi:gadgets}
\end{figure}

Now, connect the clause gadgets with a tree structure corresponding to the positions of clause rectangles in $\Phi$. Let $C_j$ be a clause rectangle in $\Phi$ with $l_1$, $l_2$, and $l_3$ corresponding to the vertical segments descending from $C_j$ from left to right. If $C_j$ is nested between vertical segments corresponding to literals $m_1$ and $m_2$ of another clause rectangle $C_k$, split the edges $(l_{j,1}, l_{j,3})$ and $(m_{k,1},m_{k,2})$ with $K$-vertices and connect the new $K$-vertices with an edge. If $C_j$ is nested under no other clause rectangle, split $(l_{j,1}, l_{j,3})$ with a $K$-vertex and connect the new vertex to $plus$ if $C_j$ corresponds to a positive clause and to $minus$ otherwise. This procedure leads to a configuration consisting of two trees of clause gadgets connected as in Fig.~\ref{fi:clause_forest}. This concludes the construction of $G$. Appendix \ref{ap:testing} proves that $G$ is $(2,2)$-planar if and only if $\Phi$ has a satisfying assignment. \qed
\end{proof}

\begin{corollary}\label{co:hardness-p-2}
The problem of computing the minimum value of $k$ such that a graph is $(k,2)$-planar is NP-hard.
\end{corollary}

\section{$(2,2)$-Planarity and $1$-Planarity} \label{se:1-planar}

%Investigating the relationship between $(k,p)$-planar graphs and well-known families of beyond planar graphs is an interesting research direction. Clearly, one can ask infinitely many questions by choosing different values of $k$ and $p$, and by choosing different families of beyond planar graphs.
\setlength\intextsep{2pt}
\setlength{\columnsep}{10pt}%
\begin{wrapfigure}{r}{0.35\textwidth}
	%	\hfill
	\small
	\centering
	\includegraphics[width=0.35\textwidth]{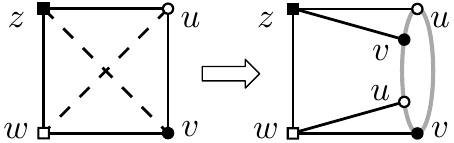}	
	\caption{Removal of a crossing in a $(2,2)$ representation.}
	\label{fi:1-planar-replacement}
\end{wrapfigure}
The NP-completeness of {\sc $(2,2)$-Planarity} suggests further investigation into the combinatorial properties of $(2,2)$-planar graphs. In this section, we study the relationship between $(2,2)$-planarity and $1$-planarity. This is partly motivated by general interest in $1$-planar graphs (see, e.g.,~\cite{DBLP:journals/csr/KobourovLM17}) and partly by the following observation. Since a $1$-planar graph admits a drawing where each edge is crossed by at most one other edge, it seems reasonable to remove each crossing of the drawing by clustering two of the vertices that are involved in the crossing as shown in Fig.~\ref{fi:1-planar-replacement}. An $n$-vertex $1$-planar graph has at most $4n-8$ edges~\cite{DBLP:journals/combinatorica/PachT97}. By Theorem~\ref{thm:edgebound2}, a $(2,2)$-planar graph with $n$ vertices has at most $4n-6$ edges, so it is not immediately clear that there are $1$-planar graphs that are not $(2,2)$-planar.

As we are going to show, however, there is an infinite family of $1$-planar graphs that are not $(2,p)$-planar for any value of $p \geq 1$. On the positive side, we demonstrate a large family of $1$-planar graphs that are $(2,2)$-planar.

\begin{theorem} \label{NIC-no(2,p)}
	For every $h>2$, there exists a $1$-planar graph with $n=5\cdot 2^h -8$ vertices and $m=18\cdot 2^h -36$ edges that is not $(2,p)$-planar, for any $p \geq 1$.
\end{theorem}
\begin{proof}
We define a recursive family of $1$-plane graphs as follows. Graph $\overline{H}_1$ consists of a single \emph{kite} $K$, which is a $1$-plane graph isomorphic to $K_4$ drawn so that all the vertices are on the boundary of the outer face. Graph $\overline{H}_i$, for $i=2,3,\dots$, has $2^i$ kites in addition to $\overline{H}_{i-1}$; these kites form a cycle in the outer face of $\overline{H}_{i-1}$, and each kite contains a vertex of the boundary of the outer face of $\overline{H}_{i-1}$ (note that $\overline{H}_{i-1}$ has $2^i$ vertices on the boundary of the outer face). See Fig.~\ref{fi:NIC-counterexample_construction} for an example. The kites of $\overline{H}_i \setminus \overline{H}_{i-1}$ are called the \emph{external kites of $\overline{H}_i$}. The embedding of $\overline{H}_i$ described in the definition will be called the \emph{canonical embedding} of $\overline{H}_i$. We also consider another possible embedding, called the \emph{reversed embedding}. Let $B$ be the boundary of the outer face in the canonical embedding of $\overline{H}_i$; in the reversed embedding of $\overline{H}_i$ the cycle $B$ is the boundary of an inner face and all the rest of the graph is embedded outside $B$. See Fig.~\ref{fi:NIC-G3-canonical} and Fig.~\ref{fi:NIC-G3-reversed} for an example. For any $h > 2$, let $\overline{H^c_h}$ be a copy of $\overline{H}_h$ with a canonical embedding, and let $\overline{H^r_h}$ be a copy of $\overline{H}_h$ with a reversed embedding. The graph obtained by identifying the external kites of $\overline{H^c_h}$ with the external kites of $\overline{H^r_h}$ is denoted as $H_h$. Fig.~\ref{fi:NIC-counterexample_labels} shows the graph $H_3$. By construction, $\overline{H}_i$ has $n_i=2^{i+1}-4$ vertices and $m_i=12\cdot 2^i-18$ edges. Hence, $H_h$ has $n=5 \cdot 2^h-8$ vertices and $m=18 \cdot 2^{h}-36$ edges.

We show that $H_h$ is not $(2,p)$-planar for any $p \geq 1$. Suppose that $H_h$ has a $(2,p)$-planar representation $\Gamma$ for some $p \geq 1$ and let $G_C$ be the graph of clusters of $H_h$. Since $\Gamma$ is planar, $G_C$ must be planar. $G_C$ can be obtained from $H_h$ by contracting each pair of vertices that is assigned to each cluster region (and removing multiple edges). Contracting a pair of vertices $u$ and $v$, the number of vertices reduces by one and the number of edges reduces by the number of paths of length at most $2$ connecting $u$ and $v$ (for each path we remove one edge).
In $H_h$, there are at most $4$ such paths between any pair of vertices. Hence, if we contract $q$ pairs of vertices, the number of vertices in $G_C$ is $n'=n-q$, while the number of edges is $m' \geq m-4q$. If $G_C$ is planar, $m' \leq 3n'-6$ and thus it must be $m-4q \leq 3(n-q)-6$, which gives $q \geq m - 3n + 6 = 3 \cdot 2^h - 6$, i.e. we must contract at least $3 \cdot 2^h - 6$ pairs of vertices. Since there are $5 \cdot 2^h - 8$ vertices, we can contract at most $\frac{5 \cdot 2^h - 8}{2}$ pairs. Thus, it must be $3 \cdot 2^{h}-6 \leq 5 \cdot 2^{h-1}-4$, i.e., $2^{h-1} \leq 2$, which can be satisfied only for $h \leq 2$.

Note that our argument is independent of the $1$-planar embedding of $H_h$. This implies that the result holds for $1$-planar graphs, not just for $1$-plane graphs.
\qed
\begin{figure}[tb]
	\centering
	\subfigure[] {\includegraphics[width=0.18\textwidth]{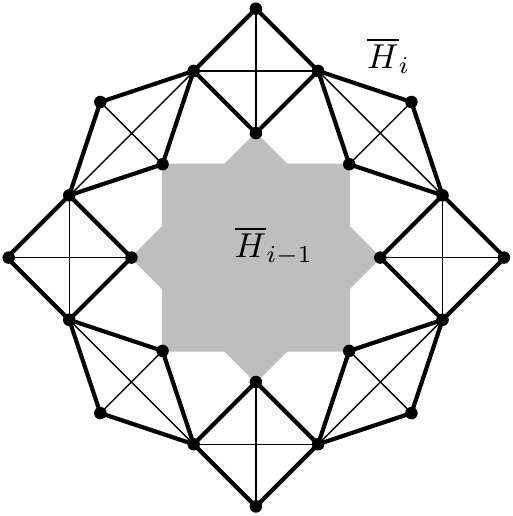} \label{fi:NIC-counterexample_construction}}
	\hfil
	\subfigure[] {\includegraphics[width=0.18\textwidth]{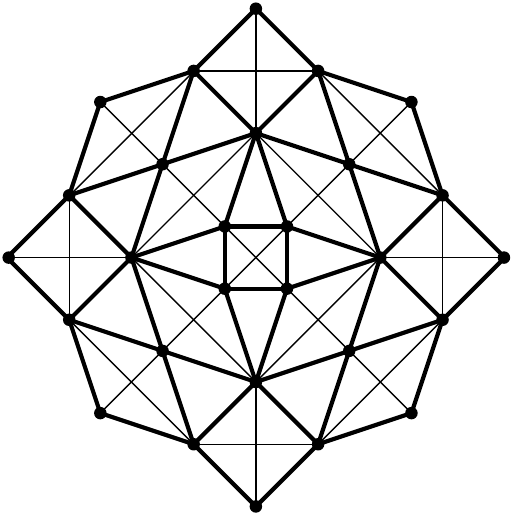} \label{fi:NIC-G3-canonical}}
	\hfil
	\subfigure[] {\includegraphics[width=0.2\textwidth]{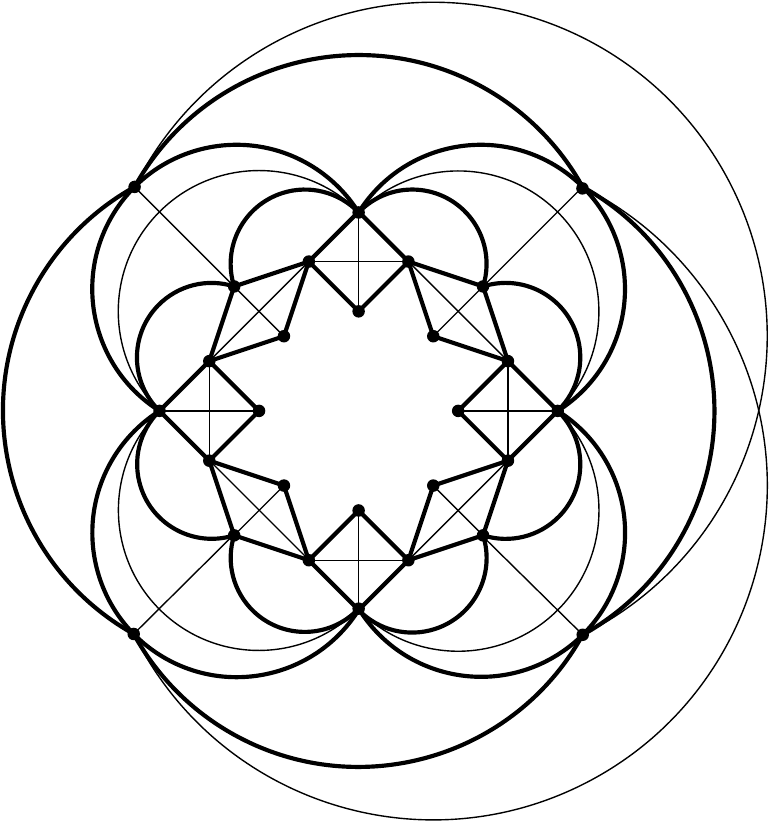} \label{fi:NIC-G3-reversed}}
	\hfil
	\subfigure[] {\includegraphics[width=0.2\textwidth]{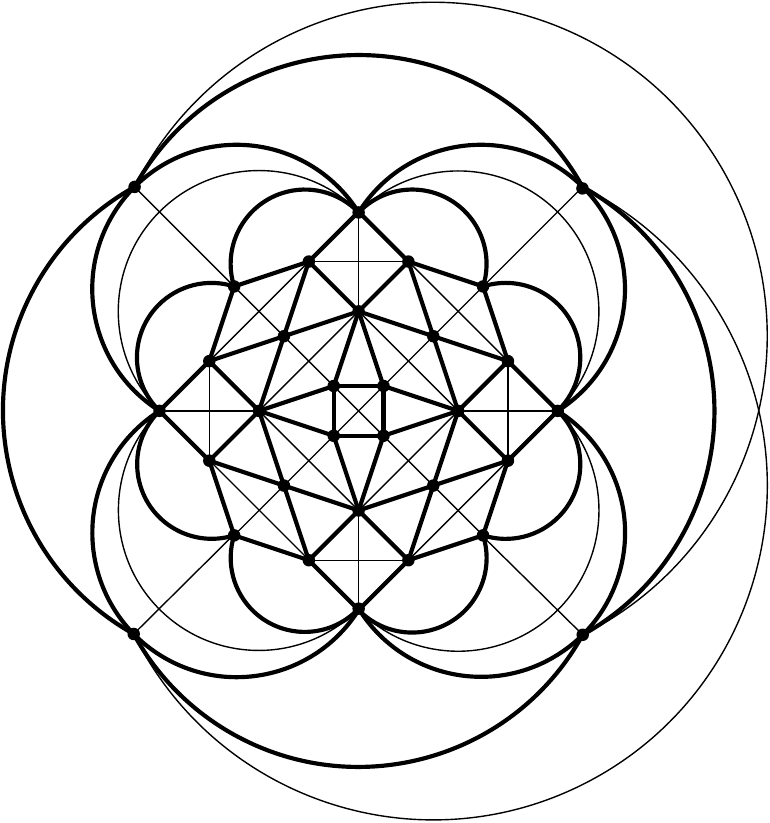} \label{fi:NIC-counterexample_labels}}
	\caption{(a) Definition of $\overline{H}_i$. (b)--(c) Canonical and reversed embedding of $\overline{H}_3$. (d) $H_3$.} \label{fi:NIC-counterexample}
\end{figure}
\end{proof}

\setlength\intextsep{-16pt}
\setlength{\columnsep}{10pt}%
\begin{wrapfigure}{r}{0.3\textwidth}
	\small
	\centering
	\includegraphics[width=0.23\textwidth]{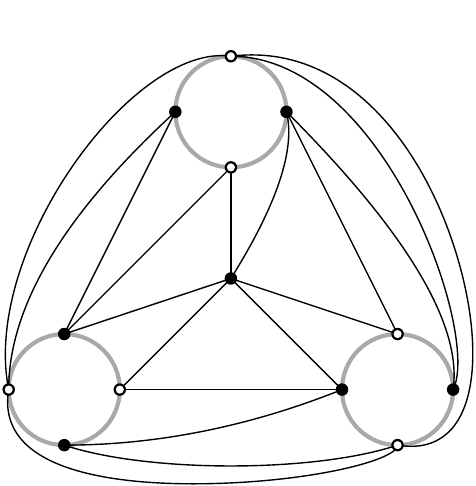}	
	\caption{A $(2,2)$-planar representation of $K_7$.}
	\label{fi:K_7-(2,2)}
\end{wrapfigure}

Theorem~\ref{NIC-no(2,p)} motivates further investigation of the relationship between $1$-planar and $(2,2)$-planar graphs. Note that there are infinitely many $(2,2)$-planar graphs that are not $1$-planar. For example, observe that every graph obtained by connecting with an edge a planar graph and $K_7$ has such a property, because $K_7$ is not $1$-planar (it has more than $4n -8 = 20$ edges) but it is $(2,2)$-planar, as depicted in Fig.~\ref{fi:K_7-(2,2)}.

In what follows, we describe a non-trivial family of $1$-planar graphs that are also $(2,2)$-planar.

Let $G$ be a $1$-plane graph, and let $e_u=(u_1,u_2)$ and $e_v=(v_1,v_2)$ be a pair of crossing edges of $G$.
Any pair $\langle u_i,v_j \rangle$, with $1 \leq i,j \leq 2$, is a \emph{representative pair} of the edge crossing defined by $e_u,e_v$. An \emph{independent set of distinct representatives} (\emph{ISDR} for short) of $G$ is a set of representative pairs such that there is exactly one representative pair per crossing and no two representative pairs in the set have a common vertex. Fig.~\ref{fi:isdr-b} shows an ISDR for the graph of Fig.~\ref{fi:isdr-a}.

 We want to show that if a $1$-plane graph $G$ has an ISDR then it is $(2,2)$-planar. The \emph{crossing edges graph} of $G$, called \emph{$ce$-graph} for short and denoted as $CE(G)$, is the subgraph of $G$ induced by the crossing pairs of $G$. $G$ is \emph{pseudoforestal} if $CE(G)$ is a pseudoforest (i.e. it has at most one cycle in each connected component). For example, the $1$-planar graph of Fig.~\ref{fi:isdr-a} is pseudoforestal, as shown in Fig.~\ref{fi:isdr-c}. The pseudoforestal $1$-planar graphs include non-trivial subfamilies of $1$-planar graphs, such as IC-planar graphs (whose $ce$-graph has maximum degree one), or the $1$-planar graphs such that each vertex is shared by at most two crossing pairs (whose $ce$-graph has maximum degree two).

\begin{figure}[tb]
	\centering
	\subfigure[] {\includegraphics[width=0.26\textwidth, page=1]{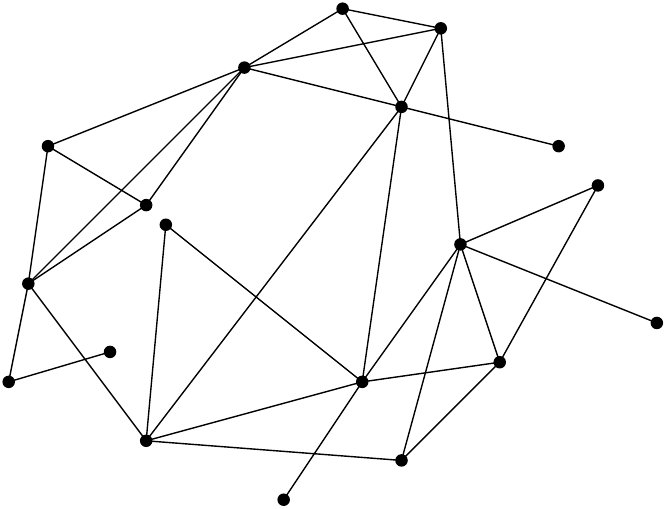} \label{fi:isdr-a}}
	\hfil
	\subfigure[] {\includegraphics[width=0.26\textwidth, page=2]{isdr} \label{fi:isdr-b}}
	\hfil
	\subfigure[] {\includegraphics[width=0.26\textwidth, page=3]{isdr} \label{fi:isdr-c}}
	\caption{(a) A $1$-planar graph $G$. (b) An ISDR of $G$. For each pair of crossing edges the representative pair is indicated with a dashed line connecting the pair. Vertices shared by different crossing pairs are replicated in each pair. (c) The $ce$-graph $CE(G)$ of $G$. } \label{fi:isdr}
\end{figure}

\begin{theorem}
  A pseudoforestal $1$-plane graph is $(2,2)$-planar.
\end{theorem}
\begin{proof}
  We start by proving that a $1$-plane graph $G$ contains an ISDR if and only if $G$ is pseudoforestal. It is known that a graph $G$ can be oriented such that the maximum in-degree is $k$ if and only if its pseudoarboricity is $k$ (i.e. the edges of $G$ can be partitioned into $k$ pseudoforests)~\cite{DBLP:conf/isaac/Kowalik06}. Thus, $G$ is pseudoforestal if and only if $CE(G)$ can be oriented so that the maximum in-degree is one. We now show that this is a necessary and sufficient condition for the existence of an ISDR $S$ in $G$. Assume that an ISDR exists. Let $e_u=(u_1,u_2)$ and $e_v=(v_1,v_2)$ be two crossing edges and let $\langle u_i,v_j \rangle$ ($1 \leq i,j \leq 2$) be the representative pair of $e_u$ and $e_v$. Direct $e_u$ towards $u_i$ and $e_v$ towards $v_j$. Doing this for each pair of crossing edges defines an orientation for all edges of $CE(G)$. In this orientation each vertex of $CE(G)$ has in-degree at most $1$, since no two pairs in $S$ share a vertex. Now suppose that $CE(G)$ has an orientation such that each vertex has in-degree at most $1$. For each pair of directed crossing edges $(u_1,u_2),(v_1,v_2)$ in $CE(G)$, we add the pair $\langle u_2,v_2 \rangle$ to $S$. Since each vertex $v$ in $CE(G)$ has in-degree at most $1$, $v$ is a vertex of at most one pair in $S$. Thus, the pairs selected for different crossing pairs are distinct and no two of them share a vertex.

 We now describe how to use an ISDR $S$ of $G$ to construct a $(2,2)$-planar representation of $G$ where each pair in $S$ is represented as a $2$-cluster that has $2$ copies for each of its vertices. Let $\Gamma$ be a $1$-planar drawing of $G$ that respects the $1$-planar embedding of $G$. Consider any two crossing edges $e_u=(u_1,u_2)$ and $e_v=(v_1,v_2)$ and denote by $c$ the point where they cross in $\Gamma$. Without loss of generality, assume that $\langle u_1,v_1 \rangle$ is the representative pair of $e_u$ and $e_v$ (see Fig.~\ref{fi:NIC-SDR-(2,2)} for an illustration). Subdivide the edge $e_u$ with a copy $v'_1$ of $v_1$ placed between $u_1$ and $c$ along $e_u$; analogously, subdivide the edge $e_v$ with a copy $u'_1$ of $u_1$. Add a curve $\lambda_1$ connecting $u'_1$ to $v'_1$ and a curve $\lambda_2$ connecting $u_1$ to $v_1$. By walking very close to the two edges $e_u$ and $e_v$, these two curves can be drawn without crossing any existing edge and so that the closed curve $\lambda$ formed by $\lambda_1$ and $\lambda_2$ together with the portion of $e_u$ from $u_1$ to $v'_1$ and the portion of $e_v$ from $v_1$ to $u'_1$ does not contain any vertex of $\Gamma$. Curve $\lambda$ defines the cluster region for the cluster containing $u$ and $v$. Replace the edge $e_u$ with a curve $\lambda_u$ connecting $u_2$ to $u'_1$ and the edge $e_v$ with a curve $\lambda_v$ connecting $v_2$ to $v'_1$. Again, by walking very close to the two edges $e_u$ and $e_v$, $\lambda_u$ and $\lambda_v$ can be drawn without crossing existing edges and without crossing each other. The replacements of $e_u$ with $\lambda_u$ and of $e_v$ with $\lambda_v$ remove the crossing between $e_u$ and $e_v$. Repeating the described procedure for every pair of crossing edges, all crossings are removed. Since for each pair of crossing edges there is a distinct representative pair and no two pair share a vertex, the result is a $(2,2)$-planar representation of $G$.\qed
\end{proof}

\begin{figure}[tb]
	\centering
	\subfigure[]{\includegraphics[width=.27\textwidth, page=5]{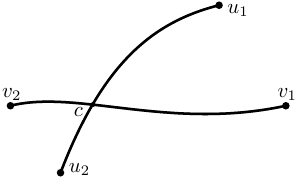} \label{fi:22-construction-a}}
	\hfil
	\subfigure[]{\includegraphics[width=.27\textwidth, page=6]{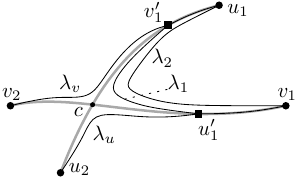} \label{fi:22-construction-b}}
	\hfil
	\subfigure[]{\includegraphics[width=.27\textwidth, page=7]{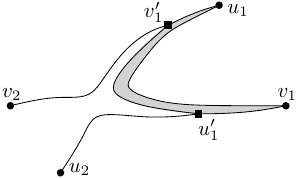} \label{fi:22-construction-c}}
	\caption{(a) Two crossing edges $e_u$ and $e_v$; (b) Construction of the cluster region and replacement of $e_u$ and $e_v$; (c) The resulting drawing.}\label{fi:NIC-SDR-(2,2)}
\end{figure}

\section{Open Problems} \label{se:open-problems}

The results in this paper suggest the following open problems: (i) Tightly bound the edge density of $(k,p)$-planar graphs for $p \geq k$; (ii) Study the complexity of $(k,p)$-planarity testing for larger values of $k$ and $p$; (iii) Further study the relationship between $1$-planar graphs and $(2,p)$-planar graphs.

\subsection*{Acknowledgements}
The authors wish to thank Maurizio Patrignani for useful discussions.

%\bibstyle{splncs04}

\bibliographystyle{splncs04}
\bibliography{biblio}

\begin{thebibliography}{10}
\providecommand{\url}[1]{\texttt{#1}}
\providecommand{\urlprefix}{URL }
\providecommand{\doi}[1]{https://doi.org/#1}

\bibitem{DBLP:conf/gd/AngeliniLBFPR15}
Angelini, P., {Da Lozzo}, G., {Di Battista}, G., Frati, F., Patrignani, M.,
  Rutter, I.: Intersection-link representations of graphs. In: {Di Giacomo},
  E., Lubiw, A. (eds.) Graph Drawing and Network Visualization, {GD} 2015.
  LNCS, vol.~9411, pp. 217--230. Springer (2015).
  \doi{10.1007/978-3-319-27261-0\_19}

\bibitem{Batagelj_Visual_2011}
Batagelj, V., Brandenburg, F., Didimo, W., Liotta, G., Palladino, P.,
  Patrignani, M.: Visual analysis of large graphs using {(X,Y)-Clustering} and
  hybrid visualizations. IEEE Transactions on Visualization and Computer
  Graphics  \textbf{17}(11),  1587--1598 (2011). \doi{10.1109/TVCG.2010.265}

\bibitem{DBLP:journals/ijcga/BergK12}
de~Berg, M., Khosravi, A.: Optimal binary space partitions for segments in the
  plane. Int. J. Comput. Geometry Appl.  \textbf{22}(3),  187--206 (2012),
  \url{http://www.worldscinet.com/doi/abs/10.1142/S0218195912500045}

\bibitem{DBLP:journals/tcs/BrandenburgDEKL16}
Brandenburg, F.J., Didimo, W., Evans, W.S., Kindermann, P., Liotta, G.,
  Montecchiani, F.: Recognizing and drawing {IC-planar} graphs. Theor. Comput.
  Sci.  \textbf{636},  1--16 (2016). \doi{10.1016/j.tcs.2016.04.026},
  \url{https://doi.org/10.1016/j.tcs.2016.04.026}

\bibitem{DBLP:conf/gd/LozzoBFP16}
{Da Lozzo}, G., {Di Battista}, G., Frati, F., Patrignani, M.: Computing
  {NodeTrix} representations of clustered graphs. In: {GD} 2016. LNCS,
  vol.~9801, pp. 107--120 (2016). \doi{10.1007/978-3-319-50106-2\_9}

\bibitem{dlpt-ntptsc-17}
{Di Giacomo}, E., Liotta, G., Patrignani, M., Tappini, A.: {NodeTrix} planarity
  testing with small clusters. In: Graph Drawing and Network Visualization -
  25th International Symposium, {GD} 2017, Boston, MA, USA, September 25-27,
  2017, Revised Selected Papers. pp. 479--491. LNCS (2017).
  \doi{10.1007/978-3-319-73915-1\_37}

\bibitem{DBLP:conf/gd/EadesM95}
Eades, P., de~Mendon{\c{c}}a~Neto, C.F.X.: Vertex splitting and tension-free
  layout. In: Graph Drawing, Symposium on Graph Drawing, {GD} '95, Passau,
  Germany, September 20-22, 1995, Proceedings. pp. 202--211 (1995).
  \doi{10.1007/BFb0021804}, \url{https://doi.org/10.1007/BFb0021804}

\bibitem{Eppstein2018}
Eppstein, D., Kindermann, P., Kobourov, S., Liotta, G., Lubiw, A., Maignan, A.,
  Mondal, D., Vosoughpour, H., Whitesides, S., Wismath, S.: On the planar split
  thickness of graphs. Algorithmica  \textbf{80}(3),  977--994 (2018).
  \doi{10.1007/s00453-017-0328-y}

\bibitem{hfm-dhvsn-07}
Henry, N., Fekete, J., McGuffin, M.J.: {NodeTrix}: a hybrid visualization of
  social networks. {IEEE} Trans. Vis. Comput. Graph.  \textbf{13}(6),
  1302--1309 (2007)

\bibitem{DBLP:journals/csr/KobourovLM17}
Kobourov, S.G., Liotta, G., Montecchiani, F.: An annotated bibliography on
  1-planarity. Computer Science Review  \textbf{25},  49--67 (2017).
  \doi{10.1016/j.cosrev.2017.06.002},
  \url{https://doi.org/10.1016/j.cosrev.2017.06.002}

\bibitem{DBLP:conf/isaac/Kowalik06}
Kowalik, L.: Approximation scheme for lowest outdegree orientation and graph
  density measures. In: Algorithms and Computation, 17th International
  Symposium, {ISAAC} 2006, Kolkata, India, December 18-20, 2006, Proceedings.
  pp. 557--566 (2006). \doi{10.1007/11940128\_56},
  \url{https://doi.org/10.1007/11940128-56}

\bibitem{DBLP:journals/combinatorica/PachT97}
Pach, J., T{\'{o}}th, G.: Graphs drawn with few crossings per edge.
  Combinatorica  \textbf{17}(3),  427--439 (1997). \doi{10.1007/BF01215922},
  \url{https://doi.org/10.1007/BF01215922}

\bibitem{Zhang2013}
Zhang, X., Liu, G.: The structure of plane graphs with independent crossings
  and its applications to coloring problems. Central European Journal of
  Mathematics  \textbf{11}(2),  308--321 (2013).
  \doi{10.2478/s11533-012-0094-7}

\end{thebibliography}

\newpage
\appendix

%%%%%%%%
%%%%%%%%
%%%%%%%%
%%%%%%%%
%%%%%%%%
\section*{Appendix}\label{ap:omitted-details}

%%%%%%%%%%%%%%%%%%%%%%%%%%%%%%%%%%%%
\section{Supplement for Proof of Theorem~\ref{thm:edgebound2}} \label{apx:edgebound2}

%In Section~\ref{se:density} we show that given a $(k,p)$-planar representation of a $(k,p)$-planar graph with $n$ vertices, the number of inter-cluster edges of $\Gamma$ is $m_{inter}\leq (kp+3)N-6$. It follows that the total number of edges is $m \leq (kp+3)N-6 + \frac{k(k-1)}{2}N$. In order to obtain the desired bound, we express the value of $N$ in terms of $n$ and $k$. If all clusters have size $k$, then $N=\frac{n}{k}$ and therefore $m \leq (kp + 3)N - 6 + \frac{k(k-1)}{2}N = n(p+\frac{3}{k}+\frac{k}{2}-\frac{1}{2})-6$.

In this section, we complete the proof of Theorem \ref{thm:edgebound2} in the case where some clusters contain fewer than $k$ vertices. Let $G$ be a $(k,p)$-planar graph, $\Gamma$ a $(k,p)$-planar representation of $G$, and $N$ the number of clusters of $G$. In Section~\ref{se:density} we showed that $m = n(p+\frac{3}{k}+\frac{k}{2}-\frac{1}{2})-6$ if all clusters contain exactly $k$ vertices.

Denote by $V_1,\dots,V_N$ the clusters of $G$ and let $k_i$ be the size of cluster $V_i$. We first add non-crossing inter-cluster edges so that the faces of $\Gamma$ external to the cluster regions are triangles. Let $\Gamma_0$ be the resulting $(k,p)$-planar representation. Notice that $\Gamma_0$ can have multiple edges. We then construct a sequence $\Gamma_0,\Gamma_1,\dots,\Gamma_N$ of $(k,p)$-planar representations so that $\Gamma_N$ has all clusters of size $k$ and each $\Gamma_i$ is obtained from $\Gamma_{i-1}$ by taking into account the cluster $V_i$. We denote by $n_i$ and $m_i$ the number of vertices and edges of $\Gamma_i$, respectively. If $k_i=k$ cluster $V_i$ is not modified and we set $\Gamma_i=\Gamma_{i-1}$. If $k_i=1$ we remove the single vertex $v$ in $V_i$ and we triangulate the face that is created by this removal (see Fig.~\ref{fi:singleton-removal-a} and~Fig.\ref{fi:singleton-removal-b}). Also in this case multiple edges can be introduced. The number of vertices of $\Gamma_i$ is then $n_i=n_{i-1}-1$, while the number of edges is $m_i=m_{i-1}-\deg(v)+\deg(v)-3=m_{i-1}-3$. If $1 < k_i < k$, we augment the cluster $V_i$ with $h_i=k-k_i$ dummy vertices, we add $p\cdot h_i$ ports in between two consecutive ports associated with two different vertices of $V_i$ (see Fig.~\ref{fi:cluster-augmentation-a} and Fig.~\ref{fi:cluster-augmentation-b}). We then add $(k-h_i)h_i + \frac{h_i(h_i-1)}{2}=h_ik-\frac{h_i^2}{2}-\frac{h_i}{2}$ edges internally to $V_i$ and $p\cdot h_i$ edges externally to $V_i$ to triangulate the face enlarged by the insertions (again multiple edges can be created). The number of vertices of $\Gamma_i$ is $n_i=n_{i-1}+h_i$, while the number of edges of $\Gamma_i$ is $m_i = m_{i-1} + ph_i + h_ik - \frac{h_i^2}{2} - \frac{h_i}{2}$. We now prove the following claim that together with the fact that $m_N \leq  n_N(p+\frac{3}{k}+\frac{k}{2}-\frac{1}{2})-6$ (because $\Gamma_N$ has all clusters of size $k$) implies that $m_0 \leq n_0(p+\frac{3}{k}+\frac{k}{2}-\frac{1}{2})-6$. Since $n=n_0$ and $m \leq m_0$, the statement follows.

\begin{figure}[b]
	\centering
	\subfigure[] {\includegraphics[width=0.23\textwidth, page= 2]{singleton-removal} \label{fi:singleton-removal-a}}
	\hfil
	\subfigure[] {\includegraphics[width=0.23\textwidth, page=4]{singleton-removal} \label{fi:singleton-removal-b}}
	\hfil
	\subfigure[] {\includegraphics[width=0.23\textwidth, page=2]{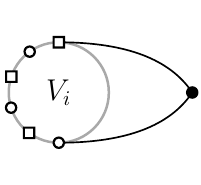} \label{fi:cluster-augmentation-a}}
	\hfil
	\subfigure[] {\includegraphics[width=0.23\textwidth, page=4]{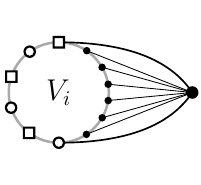} \label{fi:cluster-augmentation-b}}
	\caption{ (a) A cluster $V_i$ of size $k_i=1$, corresponding to a vertex $v$; (b) Removal of $v$ and triangulation; (c) A cluster $V_i$ of size $k_i=2$; (d) Augmentation of $V_i$ with $p \cdot h_i$ ports and triangulation of the face enlarged by the insertions.} \label{fi:clusters}
\end{figure}

	\begin{restatable}{myclaim}{claimdensity}\label{cl:one}
		If $m_i \leq  n_i(p+\frac{3}{k}+\frac{k}{2}-\frac{1}{2})-6$ then $m_{i-1} \leq  n_{i-1}(p+\frac{3}{k}+\frac{k}{2}-\frac{1}{2})-6$.
	\end{restatable}

Clearly nothing has to be proven for $k_i=k$. If $k_i=1$, we have $m_{i-1}-3 \leq (n_{i-1}-1)(p+\frac{3}{k}+\frac{k}{2}-\frac{1}{2})-6$ which gives $m_{i-1} \leq n_{i-1}(p+\frac{3}{k}+\frac{k}{2}-\frac{1}{2})-6 +3-p-\frac{3}{k}-\frac{k}{2}+\frac{1}{2}$. In order to prove that Claim~\ref{cl:one} holds in this case, we show that $3-p-\frac{3}{k}-\frac{k}{2}+\frac{1}{2} \leq 0$, which can be rewritten as $
p+\frac{3}{k}+\frac{k}{2}-\frac{7}{2} \geq 0$. Since $p \geq 1$ we have $p+\frac{3}{k}+\frac{k}{2}-\frac{7}{2} \geq \frac{3}{k}+\frac{k}{2}-\frac{5}{2} $, which is greater than or equal to $0$ for any integer value of $k$. Consider now the case $1 < k_i < k$; notice that this case is possible only for $k \geq 3$. We have $m_{i-1} + ph_i + h_ik - \frac{h_i^2}{2} - \frac{h_i}{2} \leq (n_{i-1}+h_i)(p+\frac{3}{k}+\frac{k}{2}-\frac{1}{2})-6$, which in turn gives $m_{i-1} \leq (n_{i-1}+h_i)(p+\frac{3}{k}+\frac{k}{2}-\frac{1}{2})-6 -ph_i-h_ik+\frac{h_i^2}{2}+\frac{h_i}{2}$. Again, we prove that  $h_i(p+\frac{3}{k}+\frac{k}{2}-\frac{1}{2})-ph_i - h_ik + \frac{h_i^2}{2} + \frac{h_i}{2} \leq 0$. Rearranging, we obtain $k^2-kh_i-6 \geq 0$; since $h_i \leq k-2$, we have $k^2-kh_i-6 \geq 2k-6$, which holds for every $k \geq 3$.

\section{Supplement for Proof of Theorem~\ref{th:$(3,1)$-linear}} \label{apx:31_linear}

%\theoremlinear*
%\begin{proof}
In this section, we complete the proof of Theorem \ref{th:$(3,1)$-linear} by showing that the class of $(k,1)$-planar graphs coincides with the class of planar graphs for $k=1,2,3$ and that the class of $(4,1)$-planar graphs coincides with the class of IC-planar graphs.

	If $G$ is planar, $G$ is trivially $(k,1)$-planar for all positive integers $k$. Let $G$ be a $(k,1)$-planar graph for some $k \leq 3$, and let $\Gamma$ be a $(k,1)$-planar representation of $G$. Replace each cluster of $G$ of size $h$ with an $h$-clique. Since $h \leq 3$ the obtained drawing is planar.
	
    Recall that an IC-planar graph admits a 1-planar embedding in which no two pairs of crossing edges share a vertex. Let $G$ be an IC-planar graph, and let $\Gamma$ be an IC-planar embedding of $G$. $\Gamma$ can be transformed into a $(4,1)$-planar representation of $G$ by replacing the vertices incident to each pair of crossing edges with a cluster.

    Let $G$ be a $(4,1)$-planar graph and let $\Gamma$ be a $(4,1)$-planar representation of $G$. Each cluster of $G$ is a subgraph of a $4$-clique and therefore each cluster region in $\Gamma$ can be replaced with a drawing that contains at most one pair of crossing edges. As $\Gamma$ contains no crossing inter-cluster edges, the resulting embedding is IC-planar.

\section{Proof of Lemma~\ref{le:Kvertex}}\label{apx:Kvertex}

\lemmaKvertex*
\begin{proof}
	Suppose there exists a $(2,2)$-planar representation of $G'$ that leaves $v$ unclustered or clustered with a vertex outside of $G'$. If the remaining vertices of the $G'$ subgraph are grouped into at least five clusters, $G'$ does not admit a $(2,2)$-planar representation because its graph of clusters includes a $K_5$ subgraph.
	
	Alternatively, suppose the remaining vertices of $G'$ are grouped into four clusters, in which case $G'$ consists of three $2$-clusters and two vertices which may or may not be clustered with additional vertices outside of $G'$. For the purpose of our analysis, we may ignore any vertices outside of $G'$, as their presence cannot affect the possibility of a $(2,2)$-planar representation of $G'$.

   Each $2$-cluster can contain at most $1$ intra-cluster edge, so any $(2,2)$-planar representation of $G'$ has $23$ inter-cluster edges. However, by Equation~\ref{eq:9}, we have that $m_{inter}\leq n_S + 3N-6-s$ in any $(k,p)$-planar representation $\Gamma$ of a graph $G=(V,E)$, where $s$ is the number of clusters consisting of a single vertex and $n_S$ is the total number of vertices in the skeleton of $\Gamma$. When applied to $G'$, Equation ~\ref{eq:9} implies that $23 \leq 14 + 15 - 6 - 2 = 21$, a contradiction. Thus any $(2,2)$-planar representation of $G'$ creates four $2$-clusters as shown in Fig.~\ref{fi:k8minus_22planar}. \qed
\end{proof}

\begin{figure}[tb]
	\centering
	\subfigure[]{\includegraphics[width=0.3\textwidth]{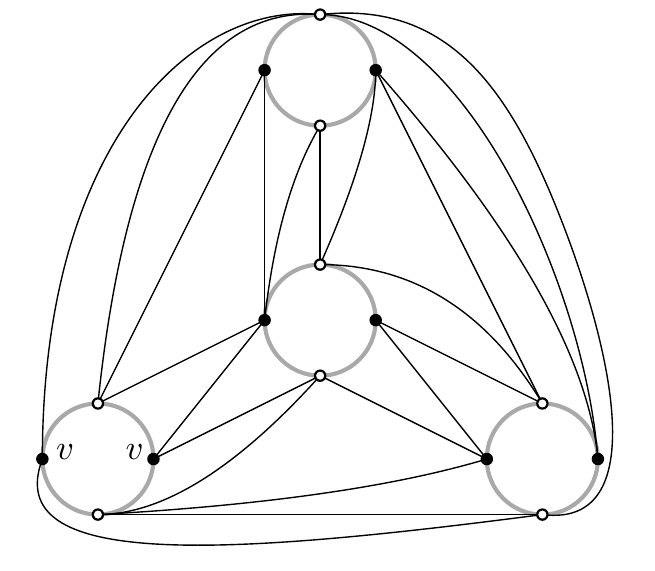}\label{fi:k8minus_22planar}}
	\hfil
	\subfigure[]{\includegraphics[width=0.5\textwidth]{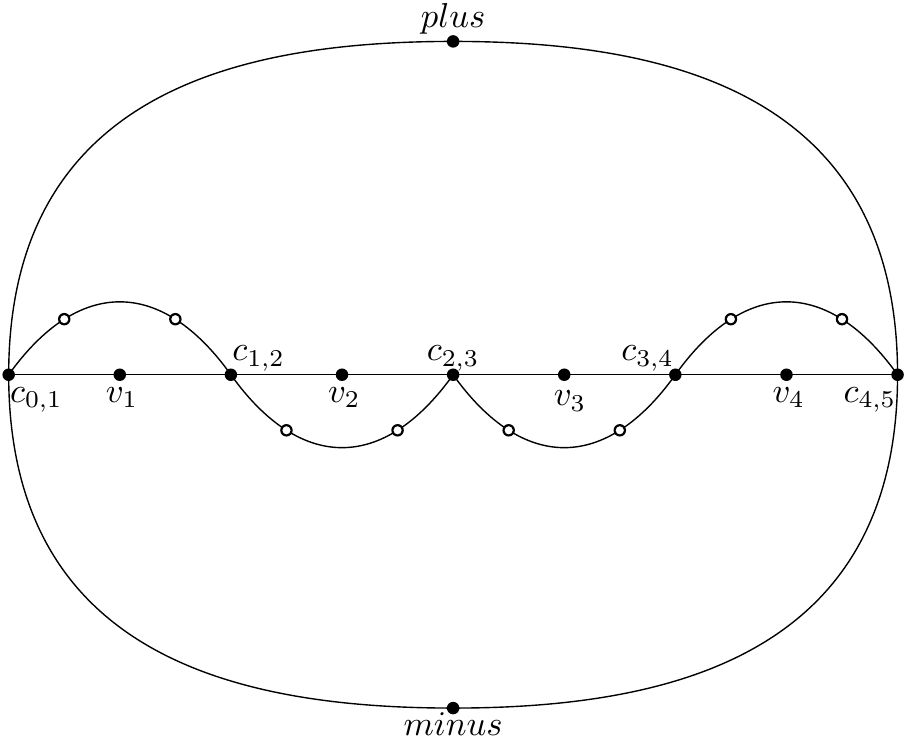}\label{fi:oriented_varcycle}}
	\caption{ (a) A $(2,2)$-planar representation of a $K$-vertex $v$ and its associated $K^-_{8}$ subgraph. (b) A drawing of the variable cycle of $G_0$ with false literal boundaries oriented according to variable assignment. }
\end{figure}

\section{Supplement for Proof of Theorem~\ref{thm:22_NP-hard}}\label{ap:testing}

In this section, we complete the proof of Theorem \ref{thm:22_NP-hard} by proving that our constructed graph $G$ is $(2,2)$-planar if and only if the corresponding instance $\Phi$ of {\sc Planar Monotone 3-SAT} is a {\sc Yes} instance.

Let $\Phi$ be a {\sc Yes} instance of {\sc Planar Monotone 3-SAT}, and let $A$ be an assignment function satisfying $\Phi$. We show that the graph $G$ corresponding to $\Phi$ is $(2,2)$-planar by constructing a $(2,2)$-planar representation of $G$ using $\Phi$ as a template.
	
	Replace each variable rectangle in $\Phi$ with the corresponding vertex of $G$ and draw the variable cycle. We refer to the region defined by the variable cycle and the $plus$ ($minus$) vertex as the \emph{positive side} (\emph{negative side}). For each variable $v_i$, draw its false literal boundary on the negative side if $A(v_i)=True$ and on the positive side if $A(v_i)=False$.	Fig.~\ref{fi:oriented_varcycle} illustrates a drawing of the variable cycle and false literal boundaries of $G_0$ according to the assignment of $v_2$ and $v_3$ to $True$ and $v_1$ and $v_4$ to $False$.
	
	Let $l_{j,i}$ be the literal vertex corresponding to clause $C_j$ and variable $v_i$. Place $l_{j,i}$ at the point of intersection between the rectangle associated with $C_j$ and the vertical segment connecting the rectangles $C_j$ and $v_i$.
	
	Connect the three literal vertices of $C_j$ to form a face, and insert $closed_j$ and $open_j$ on the interior, creating one necessary crossing. Insert the tree structure edges, which by construction can be added without creating crossings. Connect literal vertices to variable vertices, which creates a crossing on a false literal boundary precisely when the value assigned to a variable by $A$ does not match the literal. Fig.~\ref{fi:Gzero_preplanar} illustrates such a drawing of $G_0$.
	
	Resolve each crossing at a false literal boundary by clustering the literal vertex with a vertex on the boundary. The specification that each false literal boundary has at least $max(p_i, q_i)$ vertices ensures that this operation can be performed.	Because $A$ satisfies $F$, each clause gadget has at least one literal vertex that can be connected to its variable vertex without crossing a false literal boundary. Cluster this vertex with $open_j$ to resolve each clause gadget crossing. The result of this process is a $(2,2)$ representation of $G$ as illustrated in Fig.~\ref{fi:Gzero_22planar}.
	
	\begin{figure}[tb]
		\centering
		\subfigure[] {\includegraphics[width=0.47\textwidth]{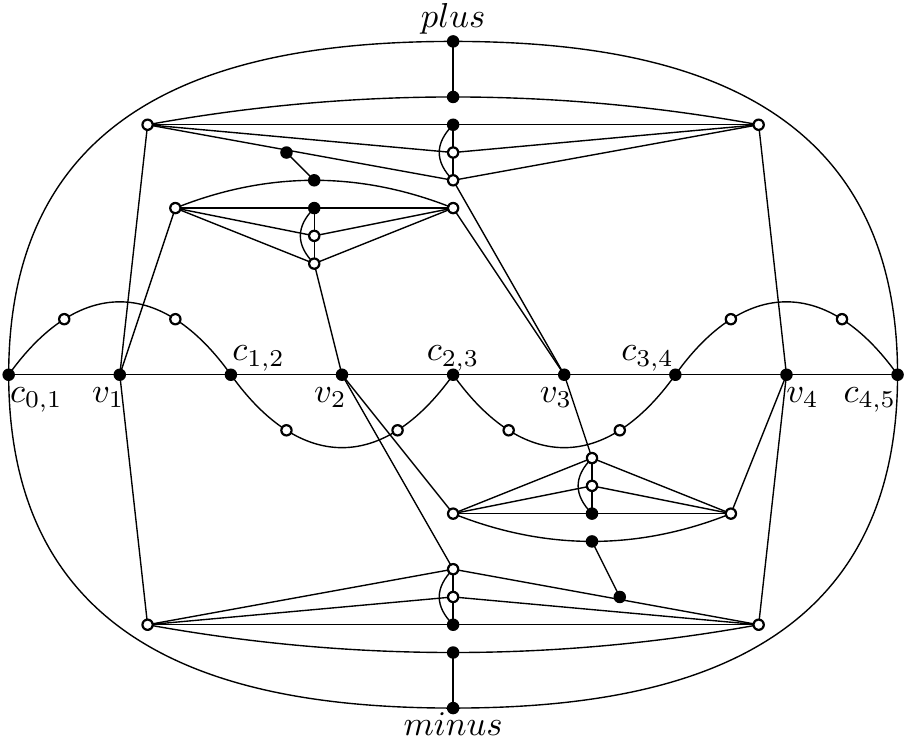} 	\label{fi:Gzero_preplanar}}
		\hfil
		\subfigure[] {\includegraphics[width=0.47\textwidth]{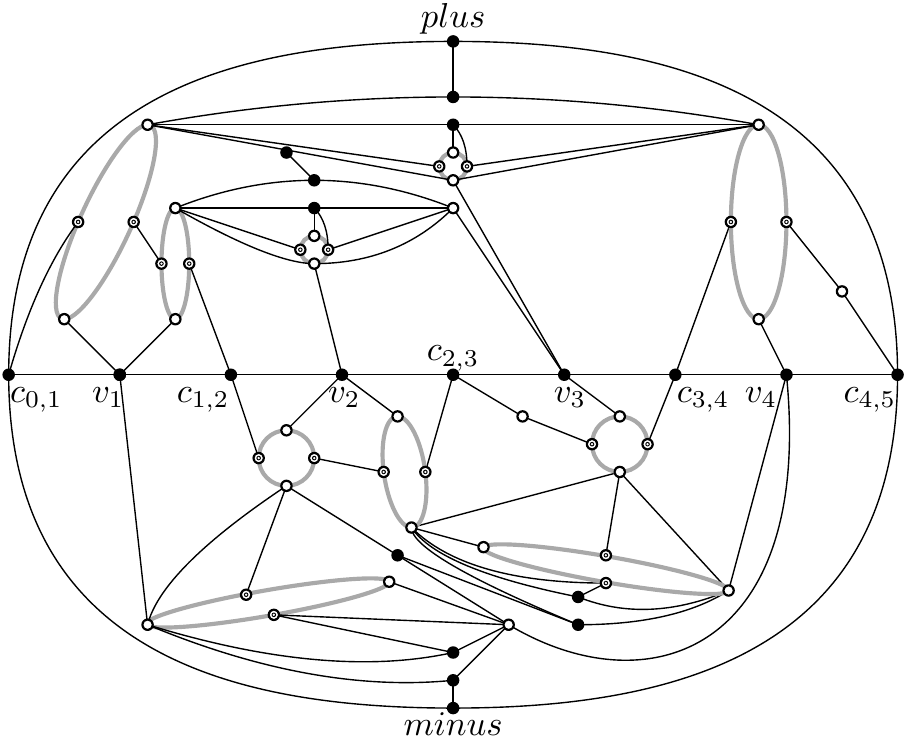} \label{fi:Gzero_22planar}}
		\caption{(a) A drawing of the graph $G_0$. (b) A $(2,2)$-planar representation of $G_0$.} \label{fi:G_0-preplanar-22planar}
	\end{figure}
	
	%	\begin{figure}[tb]
	%		\centering
	%		\includegraphics[width=0.4\textwidth]{flb_clustering}
	%		\caption{ Clustering two vertices to remove a crossing at a false literal boundary. }
	%		\label{fi:flb_clustering}
	%	\end{figure}
	
	Let $G$ be a {\sc Yes} instance of {\sc $(2,2)$-Planarity} corresponding to an instance $\Phi$ of {\sc Planar Monotone 3-SAT}. We show that $\Phi$ is a {\sc Yes} instance of {\sc Planar Monotone 3-SAT}.

	Let $\Gamma$ be a $(2,2)$-planar representation of $G$. First, note that any vertices $v_1$ and $v_2$ connected by an edge in $G$ must be drawn on the same side of the variable cycle in any $(2,2)$-planar representation of $G$. This follows from Lemma~\ref{le:Kvertex}, as neither $v_1$ nor $v_2$ can be clustered with any K-vertex in the variable cycle. Thus the positive (negative) clause gadgets must all be drawn on the same side of the variable cycle as they are connected by the tree structure to $plus$ ($minus$) and the variable vertices. We refer to the sides of the cycle with the positive and negative clause gadgets as the $positive$ and $negative$ sides of the cycle. As a consequence of Lemma~\ref{le:Kvertex}, each false literal boundary is drawn either on the positive or on the negative side of the cycle as well.
	
	Define an assignment function $A$ by setting $A(v_i)$ to $True$ ($False$) if the false literal boundary for $v_i$ is drawn on the $positive$ ($negative$) side of the vertex cycle in $\Gamma$. We claim that at least one literal vertex of each positive (negative) clause gadget is connected in $\Gamma$ to a variable vertex with $A(v_i)$ set to True (False).
	
	Without loss of generality, consider the case of a positive clause gadget $C_j$ with literals $l_{j,1}$, $l_{j,2}$, and $l_{j,3}$ connected to variables $v_1$, $v_2$, and $v_3$. Assume for contradiction that every literal vertex of $C_j$ is connected in $\Gamma$ to a variable $v$ with $A(v)=False$, which means that the false literal boundaries of $v_1$, $v_2$, and $v_3$ are drawn on the positive side of the variable cycle. We show that any placement of the K-vertex $closed_j$ creates an edge crossing in $\Gamma$, contradicting our assumption.
	
	\begin{figure}[tb]
		\centering
		\subfigure[] {\includegraphics[width=0.4\textwidth]{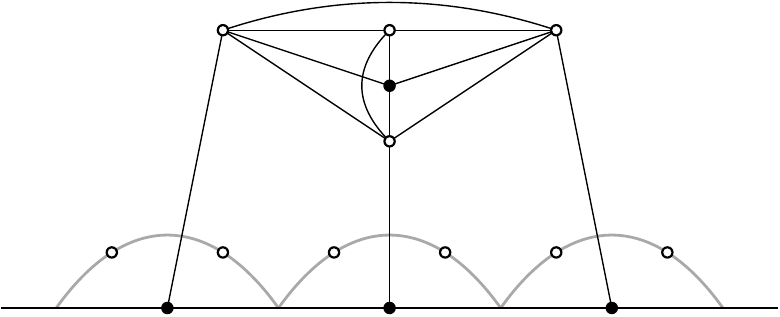} \label{fi:case1}}
		\hfil
		\subfigure[] {\includegraphics[width=0.4\textwidth]{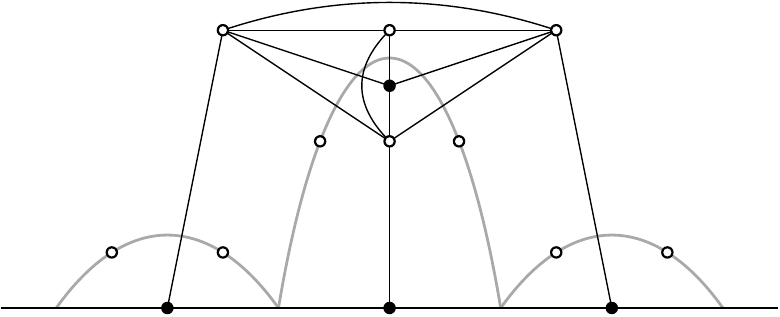} \label{fi:case2}}
		\hfil
		\subfigure[] {\includegraphics[width=0.4\textwidth]{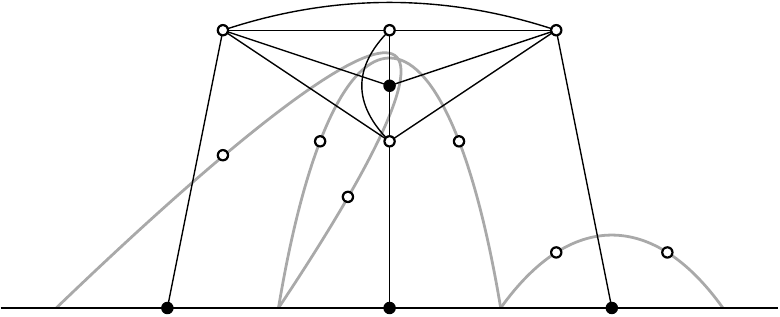} \label{fi:case3}}
		\hfil
		\subfigure[] {\includegraphics[width=0.4\textwidth]{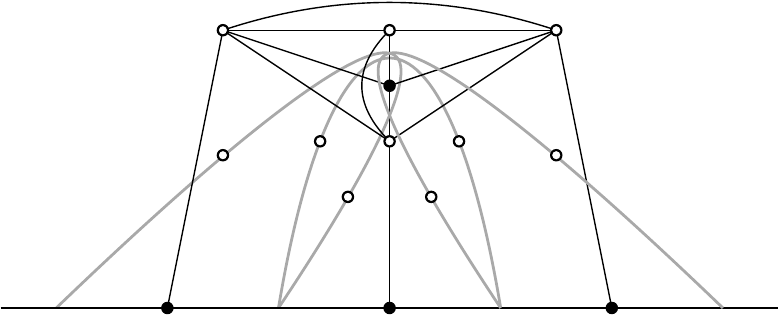} \label{fi:case4}}
		\caption{Possible placements of the clause vertex $closed_j$ relative to its three corresponding clause boundaries.} \label{fi:flbs}
	\end{figure}
	
	Suppose first that $closed_j$ is placed outside the false literal boundaries. Then each of $v_1$, $v_2$, and $v_3$ must be clustered with a boundary vertex and the clause gadget does not admit a $(2,2)$-planar representation (see Fig.~\ref{fi:case1}).
	Now suppose that $closed_j$ is drawn inside the false literal boundary of one constituent variable, $v_2$ for example. In this case, the path ($closed_j$, $l_{j,1}$, $v_1$) intersects two false literal boundaries. Because $closed_j$ and $v_2$ are $K$-vertices, only $l_{j,2}$ can be clustered with a false literal boundary vertex and thus this placement creates at least one necessary crossing (see Fig.~\ref{fi:case2}). Likewise, suppose that $closed_j$ is drawn inside the false literal boundary of two constituent variables, for example, $v_1$ and $v_2$. In this case, the path ($closed_j$, $l_{j,3}$, $v_3$) crosses three false literal boundaries and creates a necessary crossing (see Fig.~\ref{fi:case3}). Finally, suppose that $closed_j$ is drawn inside all three false literal boundaries (see Fig.~\ref{fi:case4}). In this case, the path ($closed_j$, $l_{j,1}$, $v_1$) crosses two false literal boundaries and creates a necessary crossing. Thus, regardless of the position of the a vertex $closed_j$ in $\Gamma$, at least one of the literal vertices of $C_j$ must match the assignment of its associated variable vertex. This concludes the proof of our claim, i.e., that at least one literal vertex $l_i$ of each clause gadget $C_j$ in $\Gamma$ is connected to a variable $v_i$ with $A(v_i) = l_i$. Thus $A$ is a satisfying assignment for $F$, and $\Phi$ is a {\sc Yes} instance of {\sc Planar Monotone 3-SAT}.

\end{document}